\documentclass[11pt]{article}

\usepackage[margin=1in,]{geometry}
\usepackage{todonotes}
\usepackage{hyperref}
\usepackage{xspace}
\usepackage{enumerate}
\usepackage{paralist}
\usepackage{float}
\usepackage{amsmath,amssymb,amsthm}
\usepackage{thm-restate}

\usepackage[lined,linesnumbered,vlined]{algorithm2e}

\newtheorem{theorem}{Theorem}
\newtheorem{lemma}{Lemma}[section]

\theoremstyle{definition}
\newtheorem{definition}[theorem]{Definition}

\theoremstyle{remark}

\numberwithin{equation}{section}

\newcommand{\myparagraph}[1]{\smallskip\noindent\textit{#1}}

\title{Extending Partial Representations of Circular-Arc Graphs} 

\author{
  Ji\v{r}\'{i} Fiala
	\thanks{Department of Applied Mathematics, Faculty of Mathematics and Physics, Charles University, Czech Republic
	\texttt{fiala@kam.mff.cuni.cz}. J.~Fiala was supported by the grant 19-17314J of the GA \v{C}R.}
  \and
Ignaz Rutter
\thanks{University of Passau, Germany, \texttt{rutter@fim.uni-passau.de}. Partially supported by grant Ru$\,$1903/3-1 of the German Science Foundation (DFG).}
\and
Peter Stumpf
\thanks{University of Passau, Germany, \texttt{stumpf@fim.uni-passau.de}. Partially supported by grant Ru$\,$1903/3-1 of the German Science Foundation (DFG).}
\and
Peter Zeman
	\thanks{Department of Applied Mathematics, Faculty of Mathematics and Physics, Charles University, Czech Republic
	\texttt{zeman@kam.mff.cuni.cz}. J.~Fiala was supported by the grant 19-17314J of the GA \v{C}R.}
}
\date{}

  \def\calC{{\cal C}}

  \def\calO{{\cal O}} 
 \def\calR{{\cal R}}  
   \def\calX{{\cal X}}

\def\cP{\hbox{\rm \sffamily P}}
\def\cNP{\hbox{\rm \sffamily NP}}

\def\eps{\varepsilon}
\def\O{\mathcal{O}{}}

\def\reorder{{\textsc{Reorder}}}
\def\RepExt{{\textsc{RepExt}}}
\def\Partition{{\textsc{Partition}}}

\def\int{\hbox{\rm \sffamily INT}\xspace}

\def\ca{\hbox{\rm \sffamily CA}\xspace}
\def\pca{\hbox{\rm \sffamily PCA}\xspace}
\def\npca{\hbox{\rm \sffamily NPCA}\xspace}
\def\uca{\hbox{\rm \sffamily UCA}\xspace}
\def\hca{\hbox{\rm \sffamily HCA}\xspace}
\def\nca{\hbox{\rm \sffamily NCA}\xspace}
\def\nhca{\hbox{\rm \sffamily NHCA}\xspace}
\def\phca{\hbox{\rm \sffamily PHCA}\xspace}
\def\nphca{\hbox{\rm \sffamily NPHCA}\xspace}

\def\intr{\hbox{\rm \sffamily INTR}}

\def\car{\hbox{\rm \sffamily CAR}}

\def\ucar{\hbox{\rm \sffamily UCAR}}
\def\hcar{\hbox{\rm \sffamily HCAR}}

\def\nhcar{\hbox{\rm \sffamily NHCAR}}
\def\phcar{\hbox{\rm \sffamily PHCAR}}
\def\nphcar{\hbox{\rm \sffamily NPHCAR}}

\DeclareMathOperator{\cp}{cp}
\DeclareMathOperator{\Pre}{Pre}
\DeclareMathOperator{\Reg}{Reg}

\newcommand{\qqed}{\relax}

\begin{document}

\maketitle

\begin{abstract}
The partial representation extension problem generalizes the recognition
problem for classes of graphs defined in terms of vertex representations.
We exhibit circular-arc graphs as the first example of a graph class where the
recognition is polynomially solvable while the representation extension problem
is \cNP-complete.
In this setting, several arcs are predrawn and we ask whether this partial
representation can be completed.

We complement this hardness argument with tractability results of the
representation extension problem on various subclasses of circular-arc graphs,
most notably on all variants of Helly circular-arc graphs.
In particular, we give linear-time algorithms for extending normal proper Helly
and proper Helly representations.
For normal Helly circular-arc representations we give an $\calO(n^3)$-time
algorithm.

Surprisingly, for Helly representations, the complexity hinges on the seemingly
irrelevant detail of whether the predrawn arcs have distinct or non-distinct
endpoints: 
In the former case the previous algorithm can be extended, whereas the latter
case turns out to be $\cNP$-complete.
We also prove that representation extension problem of unit circular-arc graphs
is \cNP-complete.
\end{abstract}

\section{Introduction}

An intersection representation $\calR$ of a graph $G$ is a collection
of sets $\{R(v) : v \in V(G)\}$ such that
$R(u) \cap R(v) \neq \emptyset$ if and only if $uv \in
E(G)$. Important classes of graphs are obtained by restricting the
sets $R(v)$ to some specific geometric objects.  In an \emph{interval
  representation} of a graph, each set $R(v)$ is a closed interval of
the real line; 
and in a
\emph{circular-arc representation}, the sets $R(v)$ are closed arcs of
a circle;
see Fig.~\ref{fig:cK222}.
 A graph is an \emph{interval graph} if it admits an interval representation
 and it is a \emph{circular-arc graph} if it admits a circular-arc
 representation.  We also denote the corresponding classes of graphs by \int
 and \ca, respectively.

In many cases, the availability of a geometric representation makes
computational problems tractable that are otherwise \cNP-complete.  For
example, maximum clique can be solved in polynomial time for both
interval graphs and circular-arc graphs.  Another example is the
coloring problem, which can be solved in polynomial time for interval
graphs but remains \cNP-complete for circular-arc graphs~\cite{GareyJMP80}.

A key problem in the study of geometric intersection graphs is the
\emph{recognition problem}, which asks whether a given graph has a
specific type of intersection representation.  It is a classic result
that interval graphs can be recognized in linear time.  For
circular-arc graphs the first polynomial-time recognition algorithm
was given by Tucker~\cite{tucker1980efficient}.  McConnell gave a linear-time
recognition algorithm~\cite{mcconnell2003linear}.

In this paper, we are interested in a generalization of the
recognition problem.  For a class $\calX$ of intersection representations, the
\emph{partial representation extension
  problem for $\calX$} (\RepExt{$(\calX)$} for short) is defined as follows.
  In addition to a graph $G$, the input consists of a
\emph{partial representation}~$\calR'$ that is a representation of an
induced subgraph $G'$ of $G$.  The question is whether there exists a
representation $\calR \in \calX$ of $G$ that \emph{extends} $\calR'$ in the
sense that $R(u) = R'(u)$ for all $u \in V(G')$.  The recognition
problem is the special case where the partial representation is empty.
The partial representation extension problem has been recently studied
for many different classes of intersection graphs, e.g., interval
graphs~\cite{KlavikKOSV17}, proper/unit interval
graphs~\cite{klavik2017extending}, function and permutation
graphs~\cite{klavik2012extending}, circle graphs~\cite{CFK13}, chordal
graphs~\cite{KKOS15}, and trapezoid graphs~\cite{KrawczykW17}.
Related extension problems have also been considered, e.g., for planar
topological~\cite{adfjk-tppeg-15,jkr-kttpp-13} and
straight-line~\cite{p-epsld-06} drawings, for contact
representations~\cite{chaplick2014contact}, and rectangular
duals~\cite{abs-2102-02013}.

In many cases, the key to solving the partial representation extension
problem is to understand the structure of all possible
representations.  For interval representations, the basis for this is
the characterization of Fulkerson and Gross~\cite{maximal_cliques},
which establishes a bijection between the combinatorially distinct
interval representations of a graph $G$ on the one hand and the linear
orderings~$\preceq$ of the maximal cliques of $G$ where for each
vertex $v$ the cliques containing $v$ appear consecutively
in~$\preceq$ on the other hand.  This not only forms the basis for the
linear-time algorithm using PQ-trees by Booth and
Lueker~\cite{PQ_trees}, but also shows that a PQ-tree can compactly
store the set of all possible interval representations of a graph.
The partial representation problem for interval graphs can be solved
efficiently by searching this set for one that is compatible with the
given partial representation.

Despite the fact that circular-arc graphs straightforwardly generalize
interval graphs, the structure of their representations is much less
understood.  It is not clear whether there exists a way to compactly
represent the structure of all representations of a circular-arc
graph.  There are two structural obstructions to this aim.  First, in
contrast to interval graphs, it may happen that two arcs have
disconnected intersection, namely in the case when their union covers
the entire circle.  Secondly, intervals of the real line satisfy the
\emph{Helly property}: if any pair of sets in a set system intersects,
then the intersection of the entire set system is non-empty.
Consequently, the maximal cliques of interval graphs can be associated
to distinct points of the line and also the number of maximal cliques
in an interval graph is linear in the number of its vertices.  In
contrary, arcs of a circle do not necessarily satisfy the Helly property
and indeed the number of maximal cliques can be exponential.  The
complement of a perfect matching $nK_{2}$ is an example of this
phenomenon, see Fig.~\ref{fig:cK222}b.

\begin{figure}[t]
\centering
\includegraphics[scale=1.0,page=1]{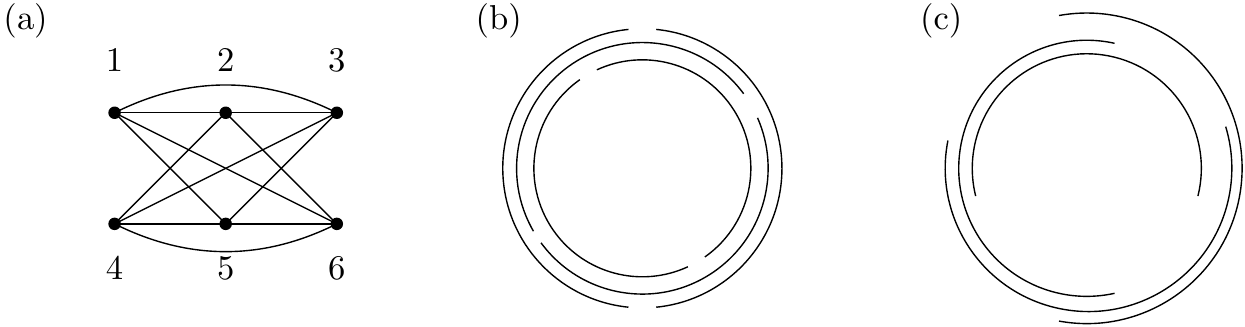}
\caption{(a) The graph $\overline{3K_{2}}$ and (b) its circular-arc representation.
(c) A non-Helly representation of $K_4$.}
\label{fig:cK222}
\end{figure}

To capture the above properties, that may have substantial impact on
explorations of  circular-arc graphs, the following specific subclasses of
circular-arc graphs have been defined and intensively
studied~\cite{Tucker1971MatrixCO,gavril1974algorithms,LinS2009survey,lin2013normal}:

\begin{compactitem}
\item \emph{Normal circular-arc graphs} ($\nca$) are circular-arc graphs that
  have an intersection representation in which the intersection of any two arcs
  is either empty or connected. 
\item \emph{Helly circular-arc graphs} ($\hca$) have an intersection
  representation that satisfies the Helly property, i.e there are no $k\ge 3$
  pairwise intersecting arcs without a point in common.
\item \emph{Proper circular-arc graphs} ($\pca$) are circular-arc graphs that
  have an intersection representation in which no arc properly contains
  another.
\item \emph{Unit circular-arc graphs} ($\uca$) are circular-arc graphs with an
  intersection representation in which every arc has a unit length. 
\end{compactitem}

The above properties can be combined together in the sense that a
single representation shall satisfy more properties simultaneously,
e.g. \emph{Proper Helly circular-arc graphs} ($\phca$) are
circular-arc graphs with an intersection representation that is both
proper and Helly~\cite{lin2007proper}. This is stronger than requiring
that a graph is a proper circular-arc graph as well as a Helly
circular-arc graph (with each property guaranteed by a different
representation), i.e., $\phca \subsetneq \pca \cap \hca$.

Analogously, since $C_4$ has a unique representation, the wheel $W_4$
is a graph with a Helly representation (the universal vertex covers
all four clique points) or a normal representation (it covers three
cliquepoints) but not normal Helly representation. Thus also
$\nhca \subsetneq \nca \cap \hca$.

Moreover, Tucker~\cite{Tucker1974} proved that every representation of a proper
(Helly) circular-arc graph that is not normal can be transformed into a normal
representation.
Hence, the following graph classes coincide $\pca=\npca$ and $\phca=\nphca$.
Fig.~\ref{fig:ca_inclusions}a shows inclusions between the defined graph classes.

We use an analogous notation for the classes of possible representations, 
i.e., 
for $X\subseteq\{\hbox{\rm \sffamily N,P,H}\}$ the symbol $X$\car{} for the
class of all $X$\ca{} representations, see Fig.~\ref{fig:ca_inclusions}b.  We
note that whether a graph $G$ with a partial representation $\cal R'$ admits an
extension depends crucially on the class of allowed representations, as
illustrated by the example of $W_4$ above.
 
\myparagraph{Our results.}
While for many classes efficient algorithms for the representation extension
problem have been found, the problem has been open for circular arc graphs for
nine years~\cite{DBLP:journals/corr/abs-1207-6960}.
We prove that $\RepExt(\car)$ is \cNP-hard. To the best of our knowledge,
it is the first known representation class for which the extension problem is
\cNP-hard while the recognition problem is in \cP.
Our reduction also works for $\RepExt(\hcar)$.

\begin{restatable}{theorem}{hcanpc}\label{thm:hca_npc}
  The problems $\RepExt(\hcar)$ and $\RepExt(\car)$ are \cNP-hard.
  $\RepExt(\car)$ is also \cNP-hard if the predrawn arcs have pairwise distinct endpoints.
\end{restatable}

We complement this result by showing tractability for several subclasses,
including all Helly variants; see Figure~\ref{fig:ca_inclusions}b. 
Linear-time algorithms for recognizing
Helly circular-arc graphs~\cite{lin2006characterizations,joeris2011linear} use
McConnell's~\cite{mcconnell2003linear} algorithm to construct a circular-arc
representation and transform it to a Helly circular-arc representation. This
cannot be exploited in the case of partial representation extension.
%


\begin{figure}[t]
\centering
\includegraphics[scale=1,page=2]{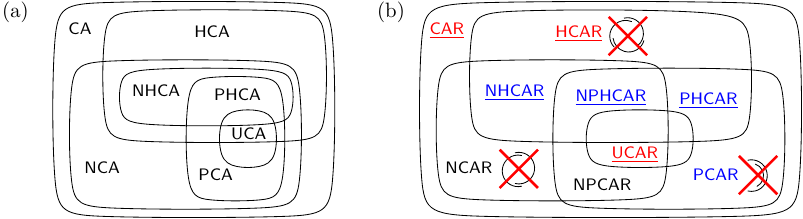}
\caption{ (a) Relationships between classes of circular-arc graphs.  (b)
  Relationships between classes of circular-arc representations. Classes
  studied in this paper are underlined. \RepExt{} is polynomial for blue, while
  \cNP-complete for red.  
}
\label{fig:ca_inclusions}
\end{figure}

Deng et al.~\cite{deng1996linear} and Lin et
al.~\cite{lin2006characterizations} characterize proper and proper Helly
circular-arc representations in terms of vertex orderings of the graph.  They
show that these orderings are unique under certain conditions.  Building on
these results, we prove the following two theorems.

\begin{restatable}{theorem}{nphcapoly}
  \label{thm:nphca_poly}
  The problem $\RepExt(\nphcar)$ can be solved in linear time.
\end{restatable}

\begin{restatable}{theorem}{phcared}
  \label{thm:phca_red}
  The problem $\RepExt(\phcar)$ can be solved in linear time.
\end{restatable}

Recall that in the case of interval graphs, PQ-trees can be used to capture all
plausible linear orderings of the maximal cliques.   Klav\'{\i}k et
al.~\cite{klavik2017extending} use this to solve $\RepExt$ for interval
representation by determining an order that is represented by the PQ-tree and
that extends a partial order that is derived from the partial representation.

The fact that Gavril~\cite{gavril1974intersection} shows that a graph $G$ is a Helly
circular-arc graph if and only if there exists a \emph{cyclic
  ordering} $\lhd$ of its maximal cliques such that for every vertex
$v$, the maximal cliques containing $v$ appear consecutively in
$\lhd$ and that Hsu and McConnell~\cite{hsu2003pc} use PC-trees to capture
all plausible cyclic orderings of the maximal cliques of a Helly
circular-arc graph makes it tempting to simply apply the same techniques to
generalize the algorithm of Klav\'ik et al.  However, this cannot be
straightforwardly applied for two reasons.  First, the clique ordering carries
little information about whether a representation is normal or not, and, even
more severely, extending a partial cyclic ordering is \cNP-complete, even
without requiring that the order be additionally represented by some given
PC-tree~\cite{gm-con-77}. We overcome this by working with suitably linearized
partial orders to show the following results.

\begin{restatable}{theorem}{nhcapoly}\label{thm:nhca_poly}
The problem $\RepExt(\nhcar)$ can be solved in $\calO(n^3)$ time.
\end{restatable}

\begin{restatable}{theorem}{hcapoly}\label{thm:hca_poly}
The problem $\RepExt(\hcar)$ can be solved in $\calO(n^3)$ time if the partial
representation consists of arcs with pairwise distinct endpoints.
\end{restatable}

It follows that Helly representations used in our reduction essentially
involve arcs that share endpoints. 
This is surprising since non-degeneracy assumptions like this are often made
without much consideration of the impact on the problem when working with graph
representations.

The bottleneck of our NHCA-algorithms is the testing of the consecutivity
constraints for all universal pairs of vertices.  A closer exploration of the
structure of the set of universal pairs may yield improvements of the running
time upper bound.

Finally, we show that involving the most tight constraints on arc lengths, the
problem becomes again computationally difficult.

\begin{restatable}{theorem}{ucanpc}\label{thm:uca_npc}
The problem $\RepExt(\ucar)$ is \cNP-complete.
\end{restatable}

The \cNP-hardness of Theorems~\ref{thm:hca_npc} and
\ref{thm:uca_npc} follows by a reduction from the $3$-$\Partition$
problem~\cite{garey1975complexity}.  For the unit case, the membership
in \cNP can be seen by a linear programming argument.

\section{Preliminaries}

\myparagraph{Cyclic order.}
Let $<\ = v_0,\dots,v_{n-1}$ and $<'\ = u_0,\dots,u_{n-1}$ be two linear orders
on a finite set $S$.
We say that $<$ and $<'$ are \emph{cyclically equivalent} if there is
$k\in\{0,\dots,n-1\}$ such that $v_i = u_{i+k}$, where the addition is modulo
$n$.
Clearly, this is an equivalence relation on the set of all linear orders on $S$.
A \emph{cyclic order} $\lhd$ on $S$ is an equivalence class of this relation.
For a linear order $<$, we denote the corresponding cyclic ordering by $[<]$.

Every linear order $<$ on $S$ induces a linear order $<'$ on a subset
$S'\subseteq S$ by omitting all ordered pairs in which the elements of
$S\setminus S'$ occur.
In this case we say that $<$ \emph{extends} $<'$ and similarly that the cyclic
order $[<]$ \emph{extends} $[<']$.

\myparagraph{Circular-arc representations.}
For any circular-arc representation $\calR$ and each connected component $C$ 
of a graph $G$ the set $\bigcup_{v \in V(C)}R(v)$ is a connected subset of the circle.
Therefore, if $G$ is a disconnected circular-arc graph, then each connected component of
$G$ has to be an interval graph. These cases can be treated with the
correpsonding algorithms for interval graphs
of~\cite{klavik2017extending,KlavikKOSV17}. Hence without loss of generality we
restrict ourselves to connected graphs in this paper.

Let $\calR$ be a representation of a circular-arc graph $G$.
For a vertex $v$ of $G$, we call the \emph{tail} $R(v)_t$ and the \emph{head}
$R(v)_h$ the two endpoints of $R(v)$. We use the convention of traversing the
arc from the tail to the head in the clockwise direction along the circle. 
We denote such an arc as $R(v)=[R(v)_t,R(v)_h]$, and its complement
$(R(v)_h,R(v)_t)$ as $R(v)^c.$

Let $\calR$ be a Helly representation of a circular-arc graph $G$.
Denote by $\calC$ the set of maximal cliques of $G$.
We assign every maximal clique $C\in \calC$ a unique point
$\cp(C)\in\bigcap_{v\in C} R(v)$ and call it the \emph{clique-point} of $C$. 

\begin{lemma}[Gavril~\cite{gavril1974intersection}]
\label{lem:gavril_hca}
A graph $G$ is a Helly circular-arc graph if and only if there exists a cyclic
ordering $\lhd$ of its maximal cliques such that for every vertex $v$,
the maximal cliques containing $v$ appear consecutively in $\lhd$.
\end{lemma}

Note that if we distribute clique points on the circle according to a cyclic
ordering $\lhd$ of Lemma~\ref{lem:gavril_hca}, then a representation $\calR$ of
$G$ can be obtained by choosing for each vertex $v$ an arc $R(v)$ that covers
exactly the clique-points $v$ belongs to.

\myparagraph{PC-Trees and the Reordering Problem.}
A \emph{PC-tree} $T$ on a set $L$ of leaves is a tree whose inner
nodes have one of two types: \emph{P-nodes} and \emph{C-nodes}.  The
neighbors around a P-node can be permuted arbitrarily, whereas the order of the
neighbors of a C-node is fixed up to reversal.  In this way a PC-tree
represents a set of cyclic orderings of its leaf set $L$.  The
usefulness of PC-trees derives from the fact that they can represent
cyclic orderings subject to consecutivity constraints.  Namely, given
a set $L$ and sets $X_1,\dots,X_r \subseteq L$, a PC-tree that
represents precisely those cyclic orderings of $L$ where each of the subsets
$X_1,\dots,X_r$ is consecutive can be computed in
$\calO(|L| + \sum_{i=1}^r |X_i|)$ time~\cite{hsu2003pc}.
In our setting, in the spirit of Lemma~\ref{lem:gavril_hca}, the leaf set $L$
will always be the set of all maximal cliques of a Helly circular-arc graph.

We make use of the following reordering problem, which can be solved
analogously to the topological sorting of PQ-trees~\cite{klavik2017extending}.
 The input consists of a PC-tree $T$, a leaf $u$ of $T$, and a partial ordering
 (not cyclic) $<$ of the remaining leaves $L' = L\setminus\{u\}$ of $T$. The
 question is whether $T$ represents a cyclic order $<'$ which induces a linear
 extension of $<$ on $L'$. 
 If this is the case, $T$ is called \emph{compatible with $<$ with respect to $u$}.  
We denote an instance of this problem by $\reorder(T,u,<)$.

\begin{restatable}{lemma}{reorderingLinear} 
  \label{lem:reordering_linear}
An instance $\reorder(T, u, <)$ can be solved in time $\calO(\ell + c)$, where
$\ell$ is the number of leaves of $T$ and $c$ is the number of comparable pairs
in the partial ordering $<$.
\end{restatable}
For a detailed proof see Section~\ref{sec:PQtree}.
\newcommand{\reorderingLinearProof}{
\begin{proof}
We root $T$ by $u$. We represent the ordering $<$ by a digraph $D$ having $c$
edges. The algorithm reorders the nodes from the bottom to the root and modifies
$D$ by contractions. Once we finish reordering a subtree, it is never modified
later.  After reordering a subtree, the corresponding vertices in $D$ are
contracted. We process a node of $T$ when all its subtrees are finished and the
corresponding digraphs are contracted to single vertices. Note that since the
tree $T$ is now rooted, the equivalent transformations (i) and (ii) correspond
to permuting the children of a P-node and reversing the children of
a C-node.

For a P-node, we then permute its children according to any topological sort of
the subdigraph $D'$ induced by the vertices corresponding to the children of the
P-node. Note that $D'$ is acyclic since $<$ is a partial ordering. For a C-node,
there are two possible orderings of its children and we check whether one of
them is feasible. The resulting PC-tree $T'$ is compatible with $<$.
\qqed
\end{proof}}


\section{Complexity}

\begin{figure}[tb]
  \centering
  \includegraphics[page=2]{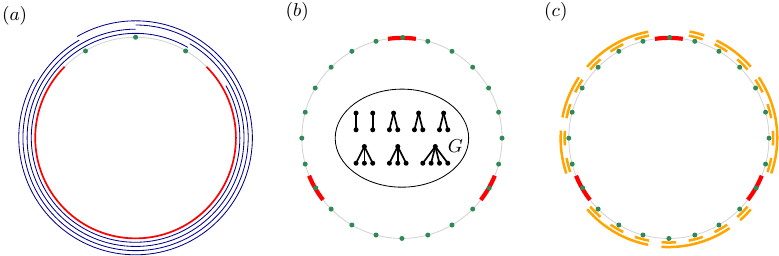}
  \caption{Illustration of the proof of Theorem~\ref{thm:hca_npc}. (a) predrawn
    universal vertices (b) the instance of $\RepExt(\car)$ obtained from
    $\{1,1,2,2,2,3,3,4\}$ of  {\sc 3-Partition}. For clarity we have omitted 24
    predrawn universal vertices and instead indicated their enpoints.   
    (c) the corresponding solution (d) variant without shared endpoints.
  }
  \label{fig:hca_npc}
\end{figure}

\begin{proof}[Sketch of proof for Theorem~\ref{thm:hca_npc}]
We first reduce the strongly \cNP-complete
problem {\sc 3-Partition}~\cite{garey1975complexity} to $\RepExt(\hcar)$.
Figure~\ref{fig:hca_npc} illustrates the proof.  Figure~\ref{fig:hca_npc}a
shows a representation with four universal vertices (blue).  The green dots
indicate the endpoints that are shared by these arcs.  The red \emph{blocker
vertex} covers all other points of the circle shared by all universal vertices.
The key insight is that each vertex that is not adjacent to the blocker vertex
must not be contained in the complement of a universal vertex and thus contain
at least one of the green points.
Figure~\ref{fig:hca_npc}b shows an instance resulting from our reduction, and
Figure~\ref{fig:hca_npc}c shows a corresponding solution from the instance
$\{1,1,2,2,2,3,3,4\}$.

For $\RepExt(\car)$, the same construction works. However, we can avoid 
shared endpoints in the partial representation with a simple modification. 
Namely, we slightly shorten the arc for each universal vertex; see
Figure~\ref{fig:hca_npc}d. 
Then we have for each former green dot a green area between the corresponding
arc ends. Each vertex that is not adjacent to the blocker vertex must now contain 
a green area. 
Note that in every solution  each leaf of a star in $G$ must contain exactly
one green area. It thus violates the Helly property with the two universal arcs
ending there.
\end{proof}

For details see Section~\ref{ssec:CARnp}.

\section{Normal Proper Helly Circular-Arc Graphs}

We show how to extend partial representations of normal proper Helly
circular-arc graphs in linear time.  To do this, we give a characterization of
all partial representations that are extendable.  This generalizes the
characterization of the extendable partial representations of proper interval
graphs~\cite[Lemma 2.4]{klavik2017extending}.
We first simplify the possible instances by reducing the number of vertices
with the same neighborhood as follows.

\myparagraph{Twin vertices.}
Recall that vertices $u,v\in V(G)$ are called twins if $N[u] = N[v]$.  It is
possible to find the equivalence classes of twin vertices in linear
time~\cite{rose1976algorithmic}.  If vertices $u$ and $v$ are twins and if
either $u$ is not predrawn or if both $u$ and $v$ are predrawn with the same
arc, then we may remove $u$ and in the final representation we can set $R(u) =
R(v)$.  This allows us to assume that each twin class either consists of a
single vertex, which is not predrawn, or it consists only of predrawn
vertices that are represented by distinct arcs.

\myparagraph{Consecutive orderings of vertices.}
Let $\lhd = [v_0, \dots, v_{n-1}]$ be a cyclic ordering of the vertices of some
graph $G$. 
We say that $\lhd$ is \emph{consecutive} if the closed neighborhood $N[v]$ of
each $v\in V(G)$ is consecutive in $\lhd$.
Note that $N[v_i]$ is consecutive in $\lhd$ if there exist positive integers
$a$ and $b$ such that $a + b \leq n$ and \[N[v_i] = \{v_{i-a}, \dots, v_{i-1},
v_i, v_{i+1}, \dots, v_{i+b}\},\] where the addition is performed modulo $n$.

Roberts characterized connected proper interval graphs in terms of consecutive linear
orderings~\cite{roberts1968representations} that are unique up to permuting
twin vertices and a complete reversal.  The following lemma, proved in slightly
different terminology by Deng et al. as Corollaries 2.7 and 2.9.
in~\cite{deng1996linear} give a characterization of proper circular-arc
graphs by cyclic ordering:

\begin{lemma}\label{lem:pca_charac}
  A graph $G$ belongs to $\pca$ if and only if $V(G)$ has a consecutive cyclic
  ordering.
\end{lemma}

They also showed that this ordering is unique up to permuting twin vertices and
a complete reversal only in the case when the complement is connected or
nonbipartite. For example, the complement of a perfect matching
$K_{2,2,\dots,2}$ on at least 6 vertices allows several distinct consecutive
orderings, like $[1,2,3,4,5,6]$ and $[1,3,2,4,6,5]$ for the graph depicted in
Fig.~\ref{fig:cK222}a. 
For normal proper Helly circular-arc graphs, the following strengthening of
Lemma~\ref{lem:pca_charac} follows from the results summarized in Lin et
al.~\cite{lin2013normal}:

\begin{lemma}\label{lem:nphca_unique}
Every $\nphca$ graph $G$ has a \emph{unique} consecutive cyclic ordering of
$V(G)$, up to permuting twins and a reversal. Such an ordering can be obtained
in  $O(m+n)$ time, where $n$ and $m$ stand for the number of vertices and
edges, respectively.
\end{lemma}

Note that each such consecutive cyclic ordering of vertices corresponds to the
cyclic order on tails of the arcs in some normal proper Helly representation.

\myparagraph{Characterization of extendable instances.}
Let $G \in \nphca$ and let $\calR'$ be a partial representation of $G$.
By the discussion above, we may assume that if $u$ and $v$ are twins, then both
are predrawn with $R(u)\neq R(v)$. 

A constraint specific to the representation extension of proper circular-arc
graphs is imposed by any pair of touching arcs, for which  the tail of one
coincides with the head of the other, i.e., $R'(v_i)_h = R'(v_j)_t$. In this
situation, it is impossible to place another tail or head between $R'(v_i)_h$
and $R'(v_j)_t$ -- in fact on this point of the circle -- as the resulting
representation would not be proper.

We use this fact as well as the uniqueness of the ordering of
Lemma~\ref{lem:nphca_unique} to characterize all instances that allow a
representation extension; see Section~\ref{sec:DetailsNPHCA}.

\begin{restatable}{lemma}{nphcaChar}\label{lem:nphca_char}
A partial normal proper Helly circular-arc representation $\calR'$ of a
connected graph $G$, where all twins are predrawn by distinct arcs, is
extendable if and only if $V(G)$ has a consecutive cyclic ordering $\lhd =
[v_0,\dots,v_{n-1}]$ satisfying:
\begin{compactenum}[(1)]
  \item \label{itm:nphca_prop1}
The ordering $\lhd$ extends the cyclic ordering $\lhd'$ of the predrawn
vertices corresponding to the clockwise cyclic ordering of the tails of their
predrawn arcs.
  \item \label{itm:nphca_prop2}
If $R'(v_i)$ and $R'(v_j)$, are distinct touching predrawn arcs such that
$R'(v_i)_h = R'(v_j)_t$, then
$N[v_i] = \{v_a, \dots,v_i,\dots,v_j\} \quad \text{and} \quad N[v_j] =
\{v_i,\dots,v_j,\dots,v_{b}\}$, for some non-negative integers $a, b$.
\end{compactenum}
\end{restatable}

\newcommand{\nphcaCharProof}{
\begin{proof}
Suppose that there is a normal proper Helly representation $\calR$ extending
$\calR'$.  Let $\lhd$ be the cyclic ordering $[v_0,\dots,v_{n-1}]$ of $V(G)$
induced by the clockwise cyclic ordering $[t_0,\dots,t_{n-1}]$ of tails
$t_i=R(v_i)_t$.  For convenience, we also write $h_i$ for $R(v_i)_h$.
The ordering $\lhd$ extends $\lhd'$ since $\calR'$ is contained in $\calR$, so
the condition~(\ref{itm:nphca_prop1}) is satisfied.

To verify condition~(\ref{itm:nphca_prop2}), suppose that there are two
predrawn arcs $R(v_i)$ and $R(v_j)$ such that $h_i = t_j$.
If $v_i$ is a universal vertex (adjacent to all other vertices), then the
condition on $N[v_i]$ holds trivially.
Otherwise, since $N[v_i]\neq V(G)$ and $N[v_i]$ is consecutive, $\lhd$ induces
a linear order on $N[v_i]$, where the minimum and maximum elements $v_a$ and
$v_b$ correspond to the first and the last neighbor of $v_i$ in clockwise
direction.
Since $h_i = t_j$, it must be that $b = j$. An analogous argument holds for $v_j$.

For the opposite implication, let $\lhd = [v_0,\dots,v_{n-1}]$ be a consecutive
cyclic ordering of $V(G)$ satisfying all three conditions. 
Without loss of generality we may assume that $G$ is not a complete graph, as
otherwise the problem is trivial -- we pick any predrawn arc and replicate it
for every non-predrawn vertex.
To construct the representation $\calR = \{R(v_0),\dots,R(v_{n-1})\}$ of $G$,
we need to determine the position of the endpoints of $R(v_i)$, for every $v_i$
that is not predrawn. 

First, we construct a cyclic order $\prec$ on the auxiliary set
\[A = \{t_0,\dots,t_{n-1}, h_0,\dots,h_{n-1}\}.\]
We construct $\prec$ in several steps.
Initially, we set $\prec\ = [t_0,\dots,t_{n-1}]$.
Next, we insert $h_i$ into $\prec$, for $i = 0,\dots,n-1$, in the following way.
For each $i=0,\dots,n-1$ we distinguish two cases, depending on whether $v_i$
is universal or not, see Fig.~\ref{fig:nphca_order}:

\begin{figure}[t]
\centering
\includegraphics[scale=1,page=1]{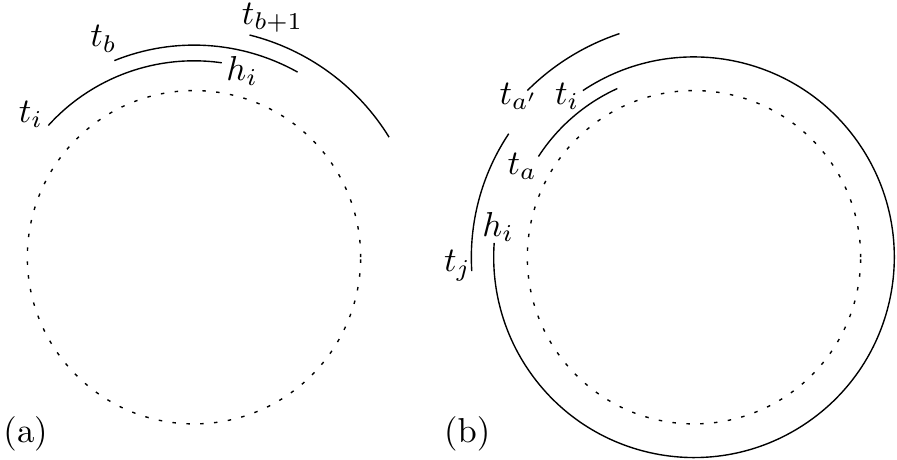}
\caption{(a) Position of the head $h_i$ in $\prec$ for a non-universal $v_i$,
 and (b) for a universal $v_i$.}
\label{fig:nphca_order}
\end{figure}

\begin{compactenum}[(a)]
\item If $v_i$ is not universal, then we determine the last neighbor $v_b$ of
  $v_i$, i.e., the vertex satisfying $N[v_i] = \{v_a, \dots,v_i,\dots,v_b\}$
  and insert $h_i$ into $\prec$ immediately before $t_{b+1}$.
\item Otherwise, i.e., when $v_i$ is universal, we determine the first
  non-universal vertex $v_a$ for which $N[v_a] = \{\dots, v_a,
    \dots,v_i,\dots\}$, i.e., the arc $R(v_a)$ shall contain the tail
    $R(v_i)_t$, and insert $h_i$ into $\prec$ immediately before $t_a$. The
    phrase \emph{the first} means that all other non-universal vertices $a'$
    which have $v_{a'}$ before $v_i$ in their neighborhood $N(v_{a'})$ are
    between $v_a$ and $v_i$ in $\lhd$.
\end{compactenum} 

We compute the order $\prec$ on the set $A$ from the cyclic vertex order $\lhd$
and from adjacencies between vertices of the graph $G$. As desribed, the order
$\prec$ depends uniquely on $\lhd$ and $G$.

The representation $\calR'$ imposes a cyclic ordering on the points
representing the tails and heads of the predrawn arcs. We extend this cyclic
ordering to the set of tails and heads of the predrawn arcs, i.e. on
$\{t_i,h_i: R'(v_i)\in \calR'\}$. In the case when a single point of the circle
is the tail and the head of touching arcs, i.e., when $R'(v_i)_h=R'(v_j)_t$, we
insert $h_i$ immediately after $t_j$. Note that as $G$ is a proper circular-ar
graph, only a single pair of arcs may touch on a single point.  Denote the
resulting cyclic ordering by $\prec'$.

Since $\calR'$ represents an induced subgraph of $G$ and as the cyclic ordering
$\lhd'$ on tails is a sub-ordering of $\lhd$ due to (1), it follows that
$\prec'$ is a sub-ordering of $\prec$. In particular, it could be obtained from
$\lhd'$ and $G'$ by the same process as $\prec$ was constructed from $\lhd$ and
$G$.

Now, we use the cyclic order $\prec$ to construct the representation $\calR$.
We identify every $t_i \in A$ with $R(v_i)_t$ and $h_i \in A$ with $R(v_i)_h$
and with a slight abuse of notation we view the elements of $A$ as the points
of the circle.

Consider now two predrawn elements $e$ and $e'$ that are consecutive in
$\prec'$ and that corresponds to distinct points of the circle --- in other
words, the interior of the arc $[e,e']$ contains no other predrawn head or
tail. 
Especially we exclude here the case when $e$ and $e'$ represent the head and
the tail of touching arcs.  For any such pair $e$ and $e'$, let
$e_0,\dots,e_{\ell-1}$ be the linear ordering of the elements in $\prec$ that
are between $e$ and $e'$ (if any). We place $e_0,\dots,e_{\ell-1}$
equidistantly between $e$ and $e'$ on the circle. This way we obtain the
representation $\calR$. By the same argument we have used for $\calR'$, the
cyclic ordering on the points representing the tails and heads can be extended
to $\lhd$ by putting heads of touching arcs immediately after the matching
tails.

Since we do not modify the endpoints of predrawn arcs, $\calR$ extends
$\calR'$.

Consider now any pair of vertices $v_i, v_j\in V(G)$.
When $v_iv_j \in E(G)$, we distinguish the following cases: 
\begin{compactitem}
  \item If both $v_i$ and $v_j$ are not universal, then we may without loss of
    generality assume that $v_j$ appears after $v_i$ in $N[v_i]= \{v_a,
    \dots,v_i,\dots, v_j, \dots, v_b\}$.  In this case the head $h_i$ has been
    inserted in $\prec$ immediately before the tail $t_{b+1}$. As the tail $t_b$
    is between $t_i$ and $t_{b+1}$ and $t_j=t_b$ or $t_j$ is between $t_i$ and
    $t_b$, it follows that the tail $t_j$ is inside the arc $R(v_i)$, see
    Fig.~\ref{fig:nphca_order}a.
  \item If $v_i$ in universal but $v_j$ not, then either $v_i$ appears before
    $v_j$ or after it in $N[v_j]$.  When $v_i$ is after $v_j$ in $N[v_j]$ we
    deduce that $t_i$ is inside the arc $R(v_j)$ as in the previous case.
    Otherwise the head $h_i$ has been inserted in $\prec$ immediately before
    $t_a$, the first non-universal vertex which has in $N[v_a]$ the vertex
    $v_j$ positioned after $v_a$. Therefore $h_j$ is between $h_i$ and $t_i$
    and the tail $t_j$ is inside the arc $R(v_i)$, see
    Fig.~\ref{fig:nphca_order}b.
  \item Finally, when both $v_i$ and $v_j$ are universal, then they are twins
    and hence are predrawn and intersecting already in $\calR'$. 
\end{compactitem}
In all three cases, the arcs $R(v_i)$ and $R(v_j)$ intersect.

On the other hand, when $v_iv_j \notin E(G)$, then none of these vertices is universal. 
In the moment of insertion of $t_i$ in $\prec$, i.e., immediately before
$t_{b+1}$ for $N[v_i]= \{v_a, \dots,v_i,\dots,v_b\}$, the tail $t_i$ appeared
between $h_i$ and $h_j$ as otherwise $N[v_i]$ would not be consecutive.
Analogously, the tail $t_j$ was inserted between $h_j$ and $h_i$.  Therefore,
the arcs $R(v_i)$ and $R(v_j)$ are disjoint.
\qqed
\end{proof}}

We are ready to prove that $\RepExt(\nphcar)$ can be solved in time $\O(n+m)$,
where~$n$ is the number of vertices and~$m$ is the number of edges of the given
graph $G$.

\nphcapoly*
\begin{proof}
The procedure is given as Algorithm~\ref{alg:nphca}.  Its correctness follows
directly from ~Lemma~\ref{lem:nphca_char}. 

\begin{algorithm}[tb]
\KwData{A graph $G$ and a partial representation $\calR'$.}
\KwResult{A \nphca{} representation $\calR$ of $G$ extending $\calR'$ or a message that it does not exist.}
\Begin{
Determine the twin classes and prune not predrawn twins\;
Find the ordering $<$  by Lemma~\ref{lem:nphca_unique}\;
Transform $<$ to $\lhd$ by reversal and permutation of predrawn twins\;
\ForAll{arcs $R'(v_i)$ and $R'(v_j)$ such that 
$R'(v_i)_h = R'(v_j)_t$}{
\lIf{$N[v_i] \cap N[v_j] \ne \{v_i,\dots,v_j\}$}{\Return{$\calR'$ has no extension.}}
}
Construct the cyclic order $(A, \prec)$ from the proof of Lemma~\ref{lem:nphca_char}\;
Convert $(A, \prec)$ into the representation $\calR$ as descibed in  Lemma~\ref{lem:nphca_char}\;
Replicate the pruned twins into $\calR$ from the non-pruned ones\;
\Return{$\calR$}
}

\caption{The algorithm for the $\RepExt(\nphcar)$ problem.}\label{alg:nphca}
\end{algorithm}

For the computational complexity note that:
\begin{compactitem}
\item Step 3 could be achieved in $O(n+m)$ time by Lemma~\ref{lem:nphca_unique}. 
\item Step 4. To check whether reversal is necessary it suffices to check two consecutive vetices belonging to discinct twin classes.
The correct ordering of the predrawn twins can be achieved by a single traversal of $<$.
\item
The for-loop at lines 5--7 has at most $n$ iterations since we check the head of every predrawn arc exactly once.
\item
The test at line $6$ can be done in constant time since it suffices to check whether $v_{i-1}v_j \in E(G)$ or $v_{i}v_{j+1} \in E(G)$.
\item
Step 8.
Following the proof of Lemma~\ref{lem:nphca_char}, we perform a single
traversal of the initial order $[t_0,\dots, t_{n-1}]$ and insert the elements
$h_0,\dots,h_{n-1}$ appropriately according to $\lhd$.
\item
For step 9, the positions of heads and tails of non-predrawn arcs can be determined in $\calO(n)$ by a single traversal of $(A, \prec)$.
\item
At line 10, each pruned twin requires only constant time to be replicated.
\end{compactitem}
\qqed
\end{proof}


For Theorem~\ref{thm:phca_red}, which deals with $\RepExt(\phca)$, note that
even though all proper Helly circular-arc graphs allow also a normal proper
representation, we cannot reduce $\RepExt(\phcar)$ to $\RepExt(\nphcar)$
directly, as the given partial representation need not to be normal.

However, the existence of a proper Helly representation extending a pair of arcs in a
not-normal position imposes strong conditions on the structure of $G$: after
pruning all universal vertices only two disjoint cliques remain.
Such instances can be solved in linear time. For details see
Section~\ref{sec:PHCA}.


\section{Normal Helly Circular-Arc Graphs}
\label{sec:NHCA}

With the following lemma, we can avoid universal vertices since they allow 
considering instances of $\RepExt(\nhcar)$ as instances for the interval case. 
\begin{lemma}
  Let $G$ be a graph with a universal vertex $u$. Then for every normal Helly
  circular-arc representation of $G$ there exists a point on the circle that is
  contained in no arc.
  \label{lem:nhcaGuniversal}
\end{lemma}
\begin{proof}
  Let $\calR$ be a normal Helly circular-arc representation of $G$.
  Assume that every point of the circle is contained in some arc. 
  Without loss of generality, we may assume that $R(u)$ is not strictly contained in
  any other arc of $\calR$. We consider the complement $R^c(u)$ of $R(u)$.
  Let $v_t$, $v_h$ be vertices whose arcs contain $R^c(u)_t$, $R^c(u)_h$
  respectively, and whose arcs maximize the intersection with $R^c(u)$.

  Note that, since $\calR$ is normal, neither $R(v_t)$ nor $R(v_h)$ contain $R^c(u)$.
  Assume that one of $R(v_t)$, $R(v_h)$, say $R(v_t)$ contains both endpoints of $R(u)$.
  We then have $R^c(u)\subseteq R(v_t)$ which by the maximality of $R(u)$
  implies $R(v_t)=R(u)$. 
  By the choice of $v_t$, we find in this case a point close to $R(u)_t$ that
  is not contained in any arc.

  In the other case, each of $R(v_t)$ and $R(v_h)$ contains exactly one endpoint of $R(u)$.
  In particular, we have $v_t\ne v_h$.
  Since $u$ is universal and since every point of $R^c(u)$ is contained in some arc,
  it follows that $R^c(u)\subseteq R(v_t)\cup R(v_h)$. We obtain $R(v_t)\cap
  R(v_h) \cap R^c(u)\ne \emptyset$ (recall that our arcs are closed sets).
  Due to the normal property, it follows that $R(v_t)\cap R(v_h) \cap R(u)=
  \emptyset$, which contradicts the Helly property.
  We conclude that there exists a point $p$ that is not contained in any arc.
\end{proof}

We assume for the rest of this section that our graph contains no universal vertices.

Let $G=(V,E)$ be a graph. Then two adjacent vertices $u, v\in V$ form a
\emph{universal pair} if each vertex $w\in V$ is adjacent either to $u$ or to $v$.

\begin{lemma}
  A graph $G$ without a universal vertex is a normal Helly circular-arc graph
  if and only if there exists a cyclic ordering $\lhd$ of its maximal cliques
  such that 
  \begin{compactenum}[(i)]
    \item for every vertex $v$, the maximal cliques containing $v$ appear
  consecutively in $\lhd$.
    \item for every universal pair $u$, $w$, the maximal cliques
      containing $u$ and $w$ appear consecutively in $\lhd$.
  \end{compactenum}

  \label{lem:nhcaG}
\end{lemma}
\begin{proof}
 If $G$ has a normal Helly circular-arc representation, then, due to the Helly
 property, we have for each maximal clique $C$ a clique point $\cp(C)$ where
 the arcs of all vertices in $C$ intersect.
 For each vertex $v$ and for each universal pair $u$, $w$ the corresponding
 cliques are consecutive in $\lhd$ since $R(v)$ and $R(u)\cap R(w)$ each are
 connected (the latter due to the normal property).

 Next assume that we have a cyclic ordering $\lhd$ of the maximal cliques
 of $G$ with properties (i) and (ii).
 We obtain a normal Helly circular-arc representation as follows;
 see Fig.~\ref{fig:lem16}.
 We first arrange the maximal clique points on the circle according to $\lhd$.
 Then, for each vertex $v$, we define the arc $R(v)$ as the smallest arc that
 contains exactly the clique points of the maximal cliques that contain $v$. Note
 that $R(v)$ is well-defined since $v$ is not universal.
 This defines a circular-arc representation $\calR$ of $G$ since any two
 intersecting arcs share a clique point. Moreover, the existence of the clique
 points shows the Helly property. It remains to show that $\calR$ is
 normal.

 Assume that there are two nodes $u$, $w$ such that $R(u)\cap R(w)$ is not connected.
 Then $u$, $w$ is a universal pair. Thus the cliques containing $u$ and
 $w$ are consecutive in $\lhd$ which contradicts $R(u)\cap R(w)$ being not connected.
 Hence, the representation is normal.
\end{proof}

\myparagraph{Extendable Partial Representations.}
We characterize all partial representations of a normal Helly circular-arc
graph $G$ that are extendable. Let $\calR'$ be a partial representation of $G$
and let $C$ be a maximal clique of $G$. We define $\Pre(C) = \{R'(v) : v \in
C \cap V(G')\}$ to be the \emph{predrawn arcs} corresponding to the vertices in $C$.

\begin{definition}
For a maximal clique $C\in\calC$, we define sets
$$\Reg^+(C) = \bigcap\limits_{w\in\Pre(C)} R'(w) \qquad\text{ and }\qquad \Reg^-(C) =
\bigcup\limits_{w\in \calR'\setminus \Pre(C)} R'(w).$$
The \emph{region} of $C$ is the set $\Reg(C) =
\Reg^+(C)\setminus \Reg^-(C)$.
\end{definition}

The set $\Reg^+(C)$ means the set of possible locations of $\cp(C)$ with respect to predrawn arcs, 
while $\Reg^-(C)$ means the forbidden locations.
Thus if the region $\Reg(C)$ is empty for some clique $C\in\calC$, then the
given partial representation is not extendable. We thus assume in the following
that no region is empty. The following lemmas give some useful properties that
hold for the regions of maximal cliques.

\begin{lemma}\label{lem:region_equivalence}
For maximal cliques $C$ and $D$, we have either $\Reg(C) \cap \Reg(D) =
\emptyset$, or $\Reg(C) = \Reg(D)$.
\end{lemma}

\begin{proof}
If $\Pre(C) = \Pre(D)$, then clearly $\Reg(C) = \Reg(D)$. So, we assume that $\Pre(C) \neq
\Pre(D)$. We can further assume that there exists an arc $R'\in
\Pre(C)\setminus \Pre(D)$.
Since the clique point $\cp(C)$ must be placed on the point of $R'$ and the clique point $\cp(D)$ cannot be placed on any point
of $R'$, we have that $\Reg(C)$ and $\Reg(D)$ are disjoint.
\qqed
\end{proof}

For every maximal clique $C$, we call the connected components of $\Reg(C)$
\emph{islands} and the connected components of its complement $\Reg^c(C)$ we call \emph{gaps of $C$}.
We say an island and a gap are \emph{neighboring}, if they share an endpoint
(where one end is open and the other is closed.) Note that every island has two
neighboring gaps and every gap has two neighboring islands.

Observe that if two maximal cliques $C,D\in \calC$ satisfy $\Pre(C)=\Pre(D)$
then $\Reg(C)=\Reg(D)$ by definition. In the other
case we obtain the following relationship:

\begin{lemma}
\label{lem:singleGap}
Let $C$ and $D$ be two maximal cliques with $\Pre(C)\ne\Pre(D)$.
Then $D$ has a gap $J$ with $\Reg(C)\subseteq J$.
\end{lemma}

\begin{proof}
From $\Pre(C)\ne\Pre(D)$ we obtain $\Reg(C)\subseteq \Reg^c(D)$.
Assume that there are two gaps $J_1$, $J_2$ of $D$ that contain points $j_1,
j_2\in\Reg(C)$. Let $I_1,I_2$ be the islands of $D$ neighboring $J_1$.
Let $i_1\in I_1$, $i_2\in I_2$. 
For $v\in\Pre(C)$ it holds that $j_1,j_2\in R'(v)$ and thus $i_1\in R'(v)$ or
$i_2\in R'(v)$ since $i_1$, $i_2$ separate $j_1$, $j_2$ on the circle; see
Fig.~\ref{fig:lem15}.
This implies 
$\Pre(C)\subseteq
\{v\in V(G'): i_1 \in R'(v) \lor i_2 \in R'(v)\} =
\{v\in V(G'): i_1 \in R'(v)\} \cup \{v\in V(G'): i_2 \in R'(v)\} 
= \Pre(D)$ as both sets in the union are equal to $\Pre(D)$.
Likewise, we obtain $\Pre(D)\subseteq\Pre(C)$. This contradicts $\Pre(C)\ne\Pre(D)$.
Hence, $\Reg(C)$ is contained in a single gap of $D$.
\end{proof}

This means that for any two maximal cliques $C$ and $D$, the clique point $\cp(C)$
must be placed in a given gap of $D$. We obtain additional consecutivity constraints.
For every gap $J$, we define the set $S_J = \{C\in\calC \mid \Reg(C)\cap J\ne \emptyset\} 
= \{C\in\calC \mid \Reg(C)\subseteq J\}$. Recall that we assumed that no region is empty.

\begin{lemma}
  \label{lem:gapConsecutivity}
  Let $\calR$ be a NHCA extension of $\calR'$. Let $J$ be a gap of some maximal clique $D$.
  Let $\le$ be the clique ordering derived from $\calR$. Then in $\le$ the set $S_J$
  is consecutive.
\end{lemma}
\begin{proof}
  Direct consequence of Lemma~\ref{lem:singleGap} since all clique points of
  that set must be placed on $J$ and all other clique points must be placed on
  its complement $J^c$.
\end{proof}

\begin{lemma}
  \label{lem:singleIsland}
  There exists a maximal clique $D$ with a single island.
\end{lemma}
\begin{proof}
Let $D$ be a maximal clique with $|\Pre(D)|$ maximal.
Then $\Reg^+(D)$ and $\Reg^-(D)$ are disjoint as otherwise there would
exist a point $p$ in $\Reg^+(D)\cap\Reg^-(D)$ with 
$\Pre(D)\subsetneq \{v\in V(G'): p \in R'(v)\}$.
This implies the existence of a (maximal) clique $C$ in $G$ with
$\Pre(D)\subsetneq \Pre(C)$ in contradiction to the choice of $D$. 

Since $\calR'$ is NHCA, this
yields that $\Reg(D) = \Reg^+(D)$ is connected. In other words, $D$ has a single island.
\end{proof}

For the rest of the section, let $D$ be a maximal clique with a single island. 
Let $p_D$ be a point in the middle of $\Reg(D)$
and let $\le$ be the linear order of points on the circle obtained by starting
at $p_D$. We consider a point $p$ to be \emph{on the left} of a point $q$ if we
have $p\le q$.

We define a partial order $\prec$ on the maximal cliques $\calC$ such that 
for each $C \in \calC\setminus\{D\}$ we have $D\prec C$ and 
for each $C, C' \in \calC\setminus \{D\}$ satisfying $\forall p\in
\Reg(C),q\in \Reg(C')\colon p < q$ we set $C\prec C'$.
Note that every linearization of a clique ordering of an extension of $\calR'$
that starts with $D$ extends $\prec$.
For any vertex $v$ let $M_v$ denote the set of maximal cliques containing $v$.

\begin{theorem}
  \label{the:NHCArepextChar}
  Let $G$ be a graph without universal vertices and let $\calR'$ be a partial
  normal Helly circular-arc representation of $G$.  
  There exists a normal Helly circular-arc
  representation of $G$ that extends $\calR'$ if and only if there exists a
  linear extension $<$ of $\prec$
 such that
  \begin{compactenum}
    \item  \label{itm:linCh0} For any pair of distinct maximal cliques $C\ne C'$ with $\Reg(C)=\Reg(C')$, the region $\Reg(C)$ is not a single point.
    \item \label{itm:linCh1} For every vertex $v$, the set $M_v$ is consecutive in $[<]$. 
    \item \label{itm:linCh2} For every universal pair $u$, $w$, the set $M_{u}\cap M_{w}$ is consecutive in $[<]$. 
    \item \label{itm:linCh4} For every gap $J$ of some $C \in \calC$, the set $S_J$ is consecutive in $[<]$. 
  \end{compactenum}
\end{theorem}

\begin{proof}
We first show that, if there is an NHCA-extension of $\calR'$, then these properties are satisfied.
We obtain $<$ as the linearization of a clique ordering of an extension of $\calR'$
starting with $D$ where the clique point of $D$ is $p_D$. 
By the construction of $\prec$, the order $<$ is a linear extension of $\prec$.
By Lemmas~\ref{lem:nhcaG} and \ref{lem:gapConsecutivity}, we obtain
properties~\ref{itm:linCh1},~\ref{itm:linCh2} and~\ref{itm:linCh4}.
Property~\ref{itm:linCh0} is necessary, since no two clique points can be
placed at the same point.
    
For the opposite implication, let $< = C_1,C_2,\dots,C_k$ be a linear extension of $\prec$ such that
properties~\ref{itm:linCh0},~\ref{itm:linCh1},~\ref{itm:linCh2} and~\ref{itm:linCh4} are
satisfied. We show that each $C_i\in \calC$ can be assigned its clique point $\cp(C_i)\in \Reg(C_i)$ such that
$\cp(C_j)<\cp(C_i)$ whenever $j<i$.

Let $\varepsilon>0$ be the $\frac{1}{2n+1}$-fraction of the length of the
shortest nontrivial island. This allows to draw all new endpoints at
distance at least $\varepsilon$ but still within any chosen island or side of
island $\Reg(D)$.
For $C_1=D$ we place $\cp(C_1)$ on $p_D$. In a greedy way, when the location of the clique points $\cp(C_1),\dots,\cp(C_{i-1})$ is settled, we determine the set $P$ of feasible points for $\cp(C_i)$ that is 
$P=\Reg(C_i) \cap \{p: p> \cp(C_{i-1})+\varepsilon\}$. If $P$ has minimum, we place $\cp(C_i)$ there, otherwise
we put $\cp(C_i)$ at $\inf(P)+\varepsilon$.
We argue that such choice always exists.

\begin{figure}[t]
\centering
\includegraphics[scale=1,page=1]{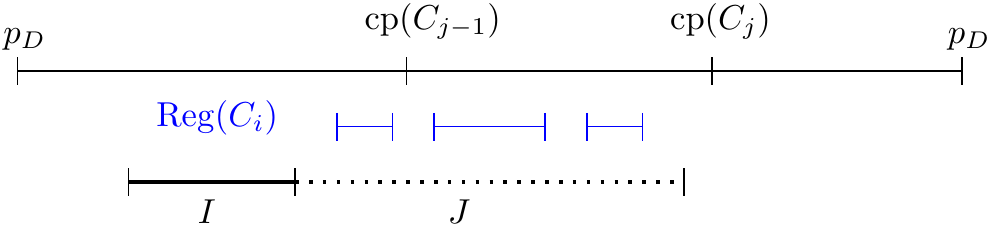}
\caption{Argument for the existence of a placement of $\cp(C_i)$.}
\label{fig:nhca_char}
\end{figure}

Assume for a contradiction that $C_i$ is the first maximal clique in the order $<$ whose 
clique point $\cp(C_i)$ cannot be properly placed. 
Note that a clique $C\ne C_i$ can only have an island $I_C$ consisting of a
single point if we have $I_C=\Reg(C)$, since $\Reg^+$ is a closed arc and all islands
separated by gaps within $\Reg^+$ have an open end.  Hence, by the choice of
$\eps$, we only place a clique $C$ at the very last point of $\Reg(C)$, if
$\Reg(C)$ consists of a single point. With property~\ref{itm:linCh0}, this can
not happen if $\Reg(C) = \Reg(C_i)$.
Therefore, one clique point must be placed to the right of $\Reg(C_i)$ before placing $\cp(C_i)$.
We identify the first maximal clique $C_j, j<i$ that is
placed to the right of all points in $\Reg(C_i)$. Since $\cp(C_j)\notin \Reg(C_i)$, 
we have that $\Pre(C_i)\ne\Pre(C_j)$.
By Lemma~\ref{lem:singleGap}, the maximal clique
$C_j$ has a gap $J$ with $\Reg(C_i)\subseteq J$, see Fig.~\ref{fig:nhca_char}.
    
Consider the neighboring island $I$ of $C_j$ to the left of $J$. Since
$\cp(C_j)$ was not placed on $I$, the clique point $\cp(C_{j-1})$ has
been placed to the right of $I$. By the choice of $C_j$, we have that
$\cp(C_{j-1})$ is not placed to the right of $J$ and thus $\cp(C_{j-1})\in J$. 
Since $\cp(C_j)$ has been placed to the right of $\Reg(C_i)$ and thus to
the right of $J$, we have $p_D\not\in J$ and thus $D\not\in S_J$.
With $D<C_{j-1}< C_j< C_i$, where $C_{j-1}, C_i\in S_J$ and $D,C_j\not\in S_J$
we get a contradiction with the property~\ref{itm:linCh4}, since $S_J$ is not consecutive in $[<]$.

From the placement of clique points, 
we obtain a NHCA-extension of $\calR'$ so that for every not yet represented 
vertex $u \notin V(G')$ we choose $R(u)$ to be the minimal arc containing 
exactly the clique points of the maximal cliques from $M_u$.

By the proof of Lemma~\ref{lem:nhcaG}, we obtain from
properties~\ref{itm:linCh2},~\ref{itm:linCh4} that this results in an NHCA representation of $G$
extending $\calR'$, since the maximal clique points are placed
correctly with regards to both predrawn arcs and new arcs, and moreover, two arcs 
intersect if and only if they are predrawn or share a maximal clique.
\end{proof}

\begin{algorithm}[t!]
\KwData{A graph $G$ and a partial representation $\calR'$.}
\KwResult{A \nhca{} representation $\calR$ of $G$ extending $\calR'$ or a message that none exists.}
\Begin{
\lIf{$G$ has a universal vertex}{resolve $\RepExt(\intr)$ instead}
\Else{
Determine the set of maximal cliques $\calC$ of $G$\;
\ForEach{$C\in \calC$}{
Determine $\Reg(C)$\;
\lIf{$\Reg(C)=\emptyset$}{\Return{$\calR'$ has no extension}}
}
$\varepsilon:=\frac{1}{2n+1}\min\{|I|,I \text{ is a non-trivial island}\}$\;
Find $C_1\in\calC$ with a single island\;
\lIf{such $C_1$ does not exist}{\Return{$\calR'$ has no extension}}
 Set $\cp(C_1)$ to the middle of $\Reg(C_1)$\;
Determine the partial order $\prec$\;
Build a PC tree $T$ on $\calC$ capturing the constraints stated in  Theorem~\ref{the:NHCArepextChar}\;
\lIf{such $T$ does not exist}{\Return{$\calR'$ has no extension}}
Solve $\reorder(T, C_1, \prec)$ to get the order $C_1<C_2<\dots<C_k$\;
\lIf{such order $<$ does not exist}{\Return{$\calR'$ has no extension}}
\For{$i=2$ \KwTo $k$}{ 
 $P:=\Reg(C_i) \cap \{p: p> \cp(C_{i-1})+\varepsilon\}$\;
 \lIf {$\min P$ exists}{$\cp(C_i):=\min P$}
 \lElse{$\cp(C_i):=\inf(P)+\varepsilon$}
}
\lForEach{$u \notin V(G')$}{draw $R(u)$ to cover exactly $\{\cp(C): C\in M_u\}$} 
\Return{$\calR$}
}
}
\caption{The algorithm for the $\RepExt(\nhcar)$ problem.}\label{alg:nhca}
\end{algorithm}

\nhcapoly*
\begin{proof}
The correctness of Algorithm~\ref{alg:nhca} follows from the already mentioned arguments.
For the computational complexity of the more complex steps note that:
\begin{compactitem}
\item
Line 2: $\RepExt(\intr)$ can be solved in linear time~\cite{klavik2013linear}.
\item
Line 4: We run the linear time recognition algorithm for Helly circular-arc graphs~\cite{lin2006characterizations} on $G$ and read its at most $n$ maximal cliques from any of its representation.

\item
Lines 5--8: The regions can be obtained in time $\calO(n)$ by traversing the circle once.
During the traversal two consecutive predrawn endpoints specify possible islands which can be assigned to appropriate maximal cliques.
\item
Line 12: The comparable pairs of the partial order $\prec$ can also be determined during the traversal in Steps 5--8.
\item
  Line 13: The $\calO(n^2)$ constraints can be computed in $O(n^3)$ time. The
  construction of the PC-tree from~\cite{hsu2003pc} also runs in $\calO(n^3)$
  time.
\item
Line 16: The $\reorder(T, C_1, \prec)$ problem can be solved in $\calO(n^2)$ time by Lemma~\ref{lem:reordering_linear}.
\end{compactitem}
\end{proof}

In light of the hardness result from Theorem~\ref{thm:hca_npc}, it is unlikely that this
result can be generalized to $\RepExt(\hcar)$. However, the hardness result for
Helly representations crucially relies on predrawn arcs sharing endpoints.
Indeed, if all predrawn arcs have distinct endpoints, the problem can be solved
in a similar fashion.

\hcapoly*
\begin{proof}[Proof sketch]
  We characterize extendable $\hca$ instances as in  
  Theorem~\ref{the:NHCArepextChar} (see Theorem~\ref{the:HCArepextChar})
  with Lemma~\ref{lem:gavril_hca} instead of Lemma~\ref{lem:nhcaG}.  Since the
  predrawn arcs have pairwise distinct endpoints, we have no islands consisting
  of a single point.
  Note that Lemma~\ref{lem:singleIsland} no longer applies and thus the
  placement of $p_D$ cannot be chosen freely.  Instead, observe that
  every maximal clique $D$ in a Helly circular-arc representation of $\calR'$
  has a clique point $\cp(D)$ that can be chosen as~$p_D$.
  Thus, we choose an arbitrary clique as $C_1$ and apply the remaining
  procedure for every island $I$ of $C_1$, choosing $\cp(C_1)\in I$.
  In contrast to our method for $\RepExt(\nhcar)$
  we have no special procedure for universal vertices.
\end{proof}

For details see Section~\ref{sec:HCA}.



\section{Details fo PC-trees and The Reordering Problem}
\label{sec:PQtree}

 \reorderingLinear*
 \reorderingLinearProof

 \section{Detailed Hardness Proofs}
 \subsection{$\RepExt(\car)$ and $\RepExt(\hcar)$}
 \label{ssec:CARnp}

\hcanpc*
\begin{proof}
We reduce the $3$-\Partition\ problem~\cite{garey1975complexity}. 
Let $S=\{s_1,\dots, s_{3n}\}$ be its instance, i.e., a collection of $3n$
integers summing up to $nt$ for some positive integer $t$
such that each satisfies $\frac{t}4<s_i<\frac{t}2$, where the goal is to
partition $S$ into $n$ disjoint subsets whose sum is always $t$. Note that the
size constraints on $s_i$'s ensure that every subset, which sums exactly to
$t$, is a suitable triple.

From $S$ and $t$ we construct a graph $G$ as follows: For the set of vertices
we choose
$V(G)=\{u_1,\dots,u_{(t+1)n},v_1,\dots,v_n,w_1,\dots,w_{tn},z_1,\dots,z_{3n}\}$,
where $u_1,\dots,u_{(t+1)n}$ are universal vertices; vertices
$v_1,\dots,v_n,w_1,\dots,w_{tn}$ form an independent set; and each $z_i$ is
connected to all of $u_1,\dots,u_{(t+1)n}$ and    $s_i$ private vertices among
$w_1,\dots,w_{tn}$.


Observe that this graph has exactly $(t+1)n$ maximal cliques, each identified
by a vertex from $v_1,\dots,v_n,w_1,\dots,w_{tn}$. Moreover, each $z_i$ belongs
to exactly $s_i$ maximal cliques.

Let $p_1,\dots, p_{(t+1)n}$ be distinct points of a circle ordered in a
clockwise direction.
For each $u_i$ let $R'(u_i)=[p_i,p_{i-1}]$. For each $v_i$, let
$R'(v_i)=[p_{(t+1)i},p_{(t+1)i}]$, (alternatively these degenerate intervals
could be extended to nondegenerate by expanding them by $\frac{1}{3}$ distance
of two nearest points $p_i,p_{i+1}$). The partial representation $\calR'$ is
the just described representation of the subgraph induced by
$\{u_1,\dots,u_{(t+1)n},v_1,\dots,v_n\}$, see Fig.~\ref{fig:hca_npc} for an
example.

When $G$ allows a representation $\calR$ extending $\calR'$, then each maximal
clique with some $w_i$ will be placed on some point $p_j$ where $j\not\equiv 0
\pmod{t+1}$. Moreover clique points corresponding to some $z_i$ must be
consecutive and do not correspond to any other $z_{i'}$ or any $v_j$.

Therefore the set system $\{\{s_i:
R(z_i)\cap\{p_{(t+1)j+1},\dots,p_{(t+1)j+t}\}\ne \emptyset \}:
j\in\{1,\dots,n\}\}$ yields the desired partition of $S$. A construction of the
representation from the set system is straightforward.

To argue that this reduction works also for \cNP-hardness of $\RepExt(\car)$
observe that we did not put any condition on the representation to enforce the
Helly property, it is guaranteed by the partial representation and the
construction itself. 

This \cNP-hardness reduction for $\RepExt(\car)$ can be altered to the situation 
when the predrawn arcs are required to have pairwise distinct endpoints just by 
shrinking each arc $R'(t_i)$ by moving its tail by $\frac{1}{5}$ distance of
the two nearest points $p_i,p_{i+1}$. The argument is the same, only to
guarantee that each $w_i$ is adjacent to all universal vertices $u_1,\dots,
u_{(t+1)n}$, the corresponding arc $R(w_i)$ would have to cover not only a
single point $p_j$ for some $j$, but instead the whole arc to the nearest tail
in the clockwise direction, formally the arc $[p_j,R'(u_j)_t]$. Observe that by
this modofication we loose the Helly property as the universal vertices share
no common point.
\end{proof}

\subsection{$\RepExt(\ucar)$}

We first argue the \cNP-hardness of Theorem~\ref{thm:uca_npc} by showing a
reduction from the $3$-$\Partition$ problem defined already in the proof of
Theorem~\ref{thm:hca_npc}.

Then we provide a linear programming argument for the membership in the class
\cNP.

\ucanpc*
\begin{proof}
For a given instance of $3$-$\Partition$, we construct a unit circular-arc
graph $G$ and its partial representation $\calR'$. For technical reasons, we
assume that $t \geq 8$.

Let $P_{2\ell}$ be a path of length $2\ell$. There exists a unit circular-arc
representation $P_{2\ell}$ such that it spans $\ell + \eps$ units, for some
$\eps > 0$. To see this, note that $P_{2\ell}$ has two independent sets of size
$\ell$ and each of this independent sets needs at least $\ell + \eps$. Let $a,
b, c$ be positive integers such that $a + b + c = t$. It follows that the
disjoint union of $P_{2a}$, $P_{2b}$, and $P_{2c}$ has a representation such
that it spans $t + \eps$ units, for some $\eps > 0$, and therefore, it can be
fit into $t+1$ units.

\begin{figure}[t]
\centering
\includegraphics[scale=0.85]{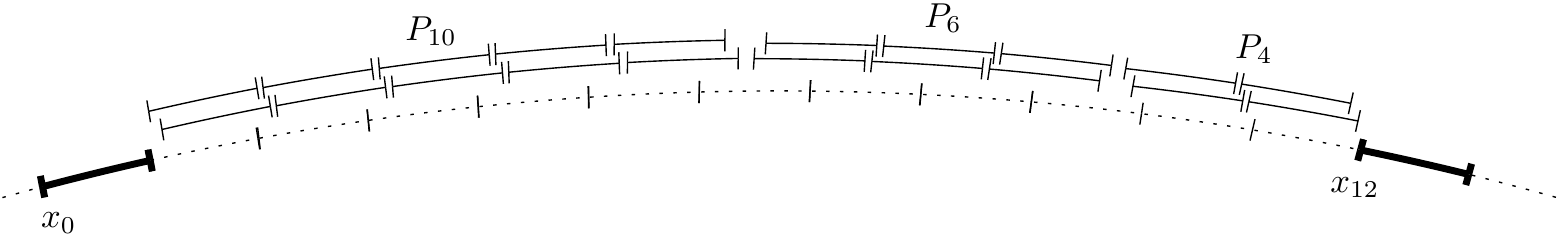}
\caption{Example of packing of three paths corresponding to the triple
$5+3+2=10$. The predrawn arcs are depicted in bold.}
\label{fig:unit_proof}
\end{figure}

Let $x_0, \dots, x_{n(t+2)-1}$ be points of the circle that divide it into
$n(t+2)$ equal parts, i.e., vertices of a regular $n(t+2)$-gon.  The graph $G$
is a disconnected graph consisting of $4n$ connected components.  For each
$s_i$, we take the path $P_{2s_i}$. We further add an isolated vertex $v_j$,
for $j = 0, \dots, n-1$. The partial representation $\calR'$ is the collection
$\{R(v_j) : j = 0,\dots, n-1\}$, where $R(v_j)$ is the arc of the circle from
$x_{j(t+2)}$ to $x_{j(t+2)+1}$ in the clockwise direction.

The predrawn arcs $R(v_0), \dots, R(v_{n-1})$ split the circle into $n$ gaps,
where each gap has exactly $t+1$ units. By the discussion above, if the $s_i$'s
can be partitioned into $n$ triples such that each triple sums to $t$, then a
representation of the disjoint union of the paths corresponding to a triple can
be placed in one of the $n$ gaps, see Fig.~\ref{fig:unit_proof}. If the partial
representation $\calR'$ can be extended, then we a have partition of the
$s_i$'s into $n$ triples such that each triple sums to $t$.

The certificate for the membership in \cNP of an instance $G$, $\calR'$ is a
linear order of $\prec$ on $V_G$ corresponding to the cyclic ordering of the
intervals together with a subset of edges $E'\subseteq E_G$.
Without loss of generality we assume that the predrawn arcs are
$R(v_1),\dots,R(v_k)$.

The following linear program, where for each $i\in\{1,\dots,n\}$ the variable
$x_i$ corresponds to the tail coordinate $R(v_i)_t$ has a feasible solution if
and only if $\calR'$ can be extended to a unit interval representation $\calR$
where its cyclic ordering of intervals (as the equivalence class) contains
$\prec$, and where the set $E'$ specifies pairs of adjacent vertices with tail
coordinates at least $\ell-1$ apart. 

$$
\begin{array}{rcll}
x_i &=& R(v_i)_t & \text{ for all } i\in\{1,\dots,k\} \\
x_i &\ge& 0 &  \text{ for all } i\in\{k+1,\dots,n\} \\
x_i &<& \ell & \text{ for all } i\in\{k+1,\dots,n\} \\
x_j - x_i &\ge& 0 & \text{ if } i\prec j \text{ and } (v_i,v_j)\in E_G \setminus E' \\
x_j - x_i &\le& 1 &\text{ if } i\prec j \text{ and } (v_i,v_j)\in E_G \setminus E' \\
x_j - x_i &\ge& \ell-1 &\text{ if } i\prec j \text{ and } (v_i,v_j)\in E' \\
x_j - x_i &>& 1 &\text{ if } i\prec j \text{ and } (v_i,v_j)\notin E_G \\
x_j - x_i &<& \ell-1 &\text{ if } i\prec j \text{ and } (v_i,v_j)\notin E_G \\
\end{array}
$$

\qqed
\end{proof}

 \section{Additional Details for Normal Proper Helly Circular Arc Graphs}
  \label{sec:DetailsNPHCA}

  \nphcaChar*
  \nphcaCharProof

  \section{Proper Helly Circular-Arc Graphs}
  \label{sec:PHCA}
Note that even though all proper Helly circular-arc graphs allow also a normal
proper Helly circular-arc representation~\cite{lin2013normal}, we cannot use
Algorithm~\ref{alg:nphca} directly, as the given partial representation may not
admit a normal extension.  We can, however, still use the same machinery.

\begin{lemma}[{\cite[Lemma~1]{lin2013normal}}]
  \label{lem:non-normal-universal}
  Let $\calR$ be a proper circular-arc representation of a graph $G$. Any two
  vertices whose arcs are in non-normal position are both universal, i.e.
  adjacent to all vertices of $G$.
\end{lemma}
\begin{proof}
  Let $u$, $v$ be two vertices with arcs in non-normal position.  Then a
  non-neighbor of $u$ would have to be represented as a proper sub-arc of
  $R(v)$ (and vice-versa) in contradiction to $\calR$ being proper.
\end{proof}

\begin{lemma}
  \label{lem:phca-nphca}
  A yes-instance of $\RepExt(\phcar)$ is a yes-instance of $\RepExt(\nphcar)$
  if and only if no two predrawn arcs are in non-normal position.
\end{lemma}

\begin{proof}
  Clearly, if the partial representation $\calR'$ contains two
  predrawn arcs in non-normal position, then there is no normal
  extension of $\calR'$.

  Conversely, assume that there is no such pair and consider a
  \phcar-extension $\calR$ of $\calR'$.
  By Lemma~\ref{lem:non-normal-universal}, only universal vertices can have arcs in
  non-normal position.
  If $\calR'$ prescribes no universal vertex, we modify~$\calR$ so
  that all universal vertices are represented by the same arc
  $R(v)$ for some arbitrary universal vertex $v$.  If $\cal R'$
  prescribes a universal vertex $v$, we modify~$\calR$ by representing
  all unprescribed universal vertices by the same arc $R(v)$.
  
  After this modification, if there still exists a non-normal pair of
  arcs, there also exists a non-normal pair of prescribed arcs.  Since
  such a pair does not exist by assumption, it follows that the
  modified representation is an $\nphcar$-extension of $\calR'$.
\end{proof}

\phcared*
\begin{proof}
  Without loss of generality we assume that $G$ is not complete as otherwise 
  an extension can always be obtained by duplication of any predrawn arc. 

  Let $\calR'$ be a partial representation of a graph $G\in \phca$.
  Without loss of generality we also assume that predrawn arcs are distinct
  as otherwise the identical arcs must correspond to twins and 
  it suffices to keep only one.

  If there is no pair of prescribed arcs in non-normal position,
  Lemma~\ref{lem:phca-nphca} implies that the problem can be solved
  with Theorem~\ref{thm:nphca_poly}.
  Hence assume that there exists a pair $R'(u),R'(v)$ of prescribed arcs in
  non-normal position that by Lemma~\ref{lem:non-normal-universal} must
  correspond to universal vertices.
  
Denote by $A$ and $B$ the two disjoint circular arcs whose union is the
intersection of $R'(u)$ and $R'(v)$ and by $C$ and $D$, resp., the arcs formed
by the points of the circle that belong only to $R'(u)$ but not $R'(v)$ and
vice-versa. The two closed arcs $A$, $B$ together with the two open $C$ and $D$
cover the whole circle, see Fig.~\ref{fig:phca}a.

\begin{figure}[t]
\centering
\includegraphics[scale=1,page=1]{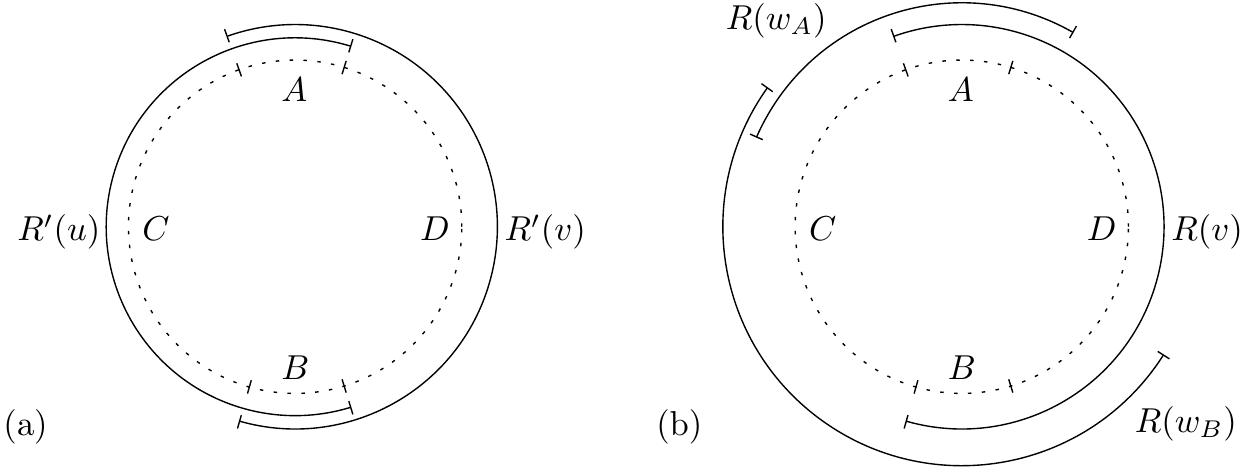}
\caption{(a) Partition of the circle into sections $A, B, C$ and $D$. (b) Violating
the Helly property by vertices for $v, w_A$ and $w_B$. }
\label{fig:phca}
\end{figure}

In order to characterize all graphs that allow a representation extension,
consider now any such graph $G$ together with its proper Helly representation
$\calR$. 

Since $\calR$ is proper, any arc $R(w)\in\calR$ distinct from $R(v)$ 
must intersect $C$ to guarantee $R(w)\not\subset R(v)$.
With an analogous condition on $D$ we get that either 
$A\subseteq R(w)$ or $B\subseteq R(w)$.

Then let $V_A=\{w_A\in V(G): A\subseteq R(w_A)\}$ and $V_B=\{w_B\in V(G): B\subseteq
R(w_B)\}$ be the sets of vertices whose arcs contain $A$ and $B$, respectively.
By the definition $V_A$ and $V_B$ are two cliques covering $V(G)$.
Since $G$ is not complete, both $V_A$ and $V_B$ contain a non-universal vertex.
Any arc intersecting $A$ and $B$ contains either $C$ or $D$, thus due to proper
and distinct positions of arcs we have $V_A \cap V_B = \{u,v\}$.

Observe that when the arc representing some $w_A\in V_A$ intersects $B$ (or vice
versa), then $w_A$ is universal --- it would intersect all vertices from $V_B$ on
$R(w_A)\cap B$. On the other hand, we claim that every arc of a universal
vertex  intersects both $A$ and $B$. Assume by a contradiction that such arc
$R(w_A)$ would intersect only $A$.
By Lemma~\ref{lem:non-normal-universal}, it is in normal position with respect
to the arc $R(w_B)$ of some non-universal vertex $w_B\in V_B$. Without loss of
generality the intersection 
of these two arcs is in $C$. Then as depicted in Fig.~\ref{fig:phca}b the arcs
of $w_A,w_B$ and $v$ violate the Helly property, a contradiction. In summary,
each universal arc  includes one of $A$ or $B$, one of $C$ or $D$ and
intersects the remaining two.

By identical arguments in slightly altered context we show that after pruning
all universal vertices from $G$ the resulting graph $G'$ is disjoint union of
two cliques $V_A'=V_A\setminus V_B$ and $V_B'=V_B\setminus V_A$.
Assume for contradiction that $G'$ contains two adjacent vertices 
$w_A\in V_A'$ and $w_B\in V_B'$.
Since $w_A$, $w_B$ are not universal, $R(w_A)$, $R(w_B)$ must be in normal position by
Lemma~\ref{lem:non-normal-universal}.
Therefore, the intersection $R(w_A)\cap R(w_B)$ is fully contained either in $C$ or in
$D$. In the first case we get a contradiction as the arcs $R(v)$, $R(w_A)$,
$R(w_B)$ violate the Helly property, see again Fig.~\ref{fig:phca}b. In the
other case, the arcs $R(u)$, $R(w_A)$ and $R(w_B)$ also violate the Helly
property. 

By Lemma~\ref{lem:non-normal-universal}, any 
arc $R(w_A)$ for a non universal vertex $w_A\in V_A'$ contains exactly one
endpoint of the each arc $R(x)$ representing a universal vertex.
The other endpoint of $R(x)$ must then belong to $R(w_B)$ of each $w_B\in V_B'$. 
This way we get a partition of the endpoints of universal arcs into two consecutive sets.
We indeed describe this partition more precisely: The arc $R(w_A)$ contains 
$\{h_i: C\subset R(u_i)\} \cup \{t_i: D\subset R(u_i)\}  $, 
while  
$\{t_i: C\subset R(u_i)\} \cup \{h_i: D\subset R(u_i)\} \subset R(w_B)$.
When $C\subset R(u_i)$ then $h_i$ belongs either to $A$ or $D$. The first case
is enforced immediately and in the latter $t_i\in B$. Hence we also have to put
$t_i$ into $R(w_A)$.

\begin{figure}[t]
\centering
\includegraphics[scale=1,page=2]{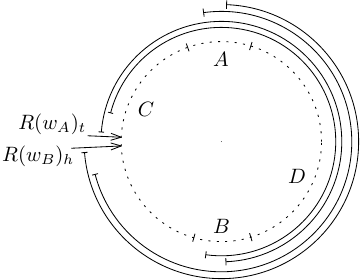}
\caption{Position of $R(w_A)_t$ and $R(w_B)_h$ with respect to 
universal arcs containig $D$.}
\label{fig:phca-2}
\end{figure}

This necessary condition is also sufficient, namely a partial representation
$\calR'$ of a non-complete graph proper-Helly circular graph $G$ with arcs
$R(u),R(v)$ in not-normal position allows an extension if and only if the sets
$\{h_i: C\subset R'(u_i)\} \cup \{t_i: D\subset R'(u_i)\}$, 
and   
$\{t_i: C\subset R'(u_i)\} \cup \{h_i: D\subset R'(u_i)\}$ are consecutive ---
if no $w_A\in V_A'$ is pre-drawn, we choose $R(w_A)$ to be the shortest arc
containing the first of these two sets within $A\cup C\cup D$ and analogously
for any $w_B\in V_B'$, see Fig.~\ref{fig:phca-2}.
The remaining arcs could be obtained by replication.

Our arguments convert directly to the following algorithm:
Accept if $G$ is complete. If the partial proper Helly
representation has no pair of arcs in non-normal position, use
Algorithm~\ref{alg:nphca}. 
Otherwise check whether $G'$ is isomorphic to the
union of two disjoint cliques; if not, reject. 
Finally, check whether the endpoints of the predrawn arcs of universal vertices
could be split into two consecutive sets described above and if the predrawn
arcs non-universal ones cover these sets appropriately. If not, reject.

In the affirmative case, a representation could be constructed by replication
of arcs of three representatives: a universal vertex, one vertex from $V_A'$
and one from $V_B'$.
A minor simplification could applied when a representative from $V_A'$ or
$V_B'$ is already pre-drawn and recognized. 
\end{proof}

\section{Illustrations for the Normal Helly Case}

\begin{figure}[H]
\centering
\includegraphics{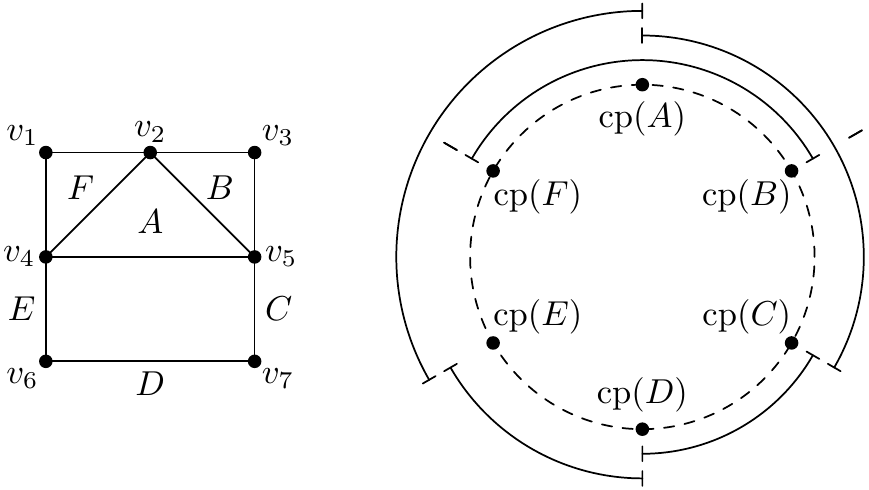}
\caption{An example of a construction of a normal Helly circular-arc representation.}
\label{fig:lem16}
\end{figure}

\begin{figure}[H]
\centering
\includegraphics{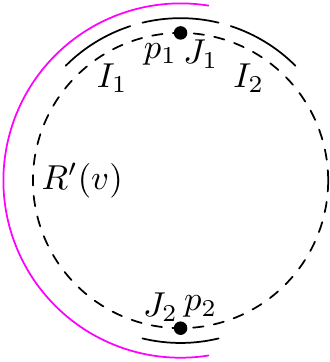}
\caption{The region of a maximal clique cannot intersect two gaps of another.}
\label{fig:lem15}
\end{figure}

\section{Helly Circular Arc Graphs}
\label{sec:HCA}

Surprisingly, when the Helly property is required, we can solve
$\RepExt(\hcar)$ similarly to $\RepExt(\nhcar)$
when the predrawn arcs have pairwise distinct endpoints. The most notable
difference is that Lemma~\ref{lem:singleIsland} no longer applies and we have
to test every island of clique $D$ for placing $p_D$.

\begin{theorem}
  \label{the:HCArepextChar}
  Let $\calR'$ be a partial $\hca$ representation of $G$ where the arcs have
  pairwise distinct endpoints and let $D$ be a maximal clique of $G$.
  There exists an $\hca$ representation of $G$ that extends
  $\calR'$ if and only if for some point $p_D\in \Reg(D)$ with the corresponding
  partial order $\prec$ as defined in Section~\ref{sec:NHCA}
  there exists a linear extension $<$ of $\prec$ such
  that:
  \begin{compactenum}
    \item \label{itm:HlinCh1} For every vertex $v$, $M_v$ is consecutive in $[<]$. %
    \item \label{itm:HlinCh4} For every gap $J$ of a $C\in \calC$, the set
      $S_J$ is consecutive in $[<]$. 
  \end{compactenum}
\end{theorem}

\begin{proof}
We first show that, if there is such an $\hca$-extension of $\calR'$, then
these properties are satisfied.

We obtain $<$ as the linearization of a clique ordering of an extension of $\calR'$
starting with $D$ where the clique point of $D$ is $p_D$. 
By the construction of $\prec$, the order $<$ is a linear extension of $\prec$.
By Lemmas~\ref{lem:gavril_hca} and \ref{lem:gapConsecutivity}, we obtain
properties~\ref{itm:HlinCh1} and~\ref{itm:linCh4}.
    
For the opposite implication, let $p_D\in \Reg(D)$ be on an island $I_1$ of $D$
and let $< = C_1,C_2,\dots,C_k$ be a linear extension of $\prec$ such that
properties~\ref{itm:HlinCh1} and~\ref{itm:HlinCh4} are
satisfied. We show that each $C_i\in \calC$ can be assigned its clique point
$\cp(C_i)\in \Reg(C_i)$ such that $\cp(C_j)<\cp(C_i)$ whenever $j<i$.  We
relocate $p_D$ to the middle of island $I_1$.  Note that the
properties~\ref{itm:HlinCh1} and~\ref{itm:HlinCh4} remain
satisfied for the new ordering $\prec$ since the only possible change is that
relations to cliques with island $I_1$ may get lost. 

Let $\varepsilon>0$ be the $\frac{1}{2n+1}$-fraction of the length of the
shortest nontrivial island. This choice allows us to draw all new endpoints at
distance at least $\varepsilon$ but still within any chosen island or side of
island $I_1$.

For $C_1 = D$ we place $\cp(C_1)$ on $p_D$. In a greedy way, when the location
of the clique points $\cp(C_1),\dots,\cp(C_{i-1})$ is settled, we determine the
set $P$ of feasible points for $\cp(C_i)$ that is $P = \Reg(C_i) \cap \{p: p>
\cp(C_{i-1})+\varepsilon\}$. If $P$ has minimum, we place $\cp(C_i)$ there,
otherwise we put $\cp(C_i)$ at $\inf(P)+\varepsilon$.
We argue that such choice always exists.

Assume for a contradiction that $C_i$ is the first maximal clique in the order $<$ whose 
clique point $\cp(C_i)$ cannot be properly placed. 
Note that a maximal clique $C\ne C_i$ can never have an island $I_C$ consisting of a
single point since this can only occur as intersection of two predrawn arcs
ending at that point. Hence, by the choice of
$\eps$, we never place a clique $C$ at the very last point of $\Reg(C)$.
Therefore, one clique point must be placed to the right of $\Reg(C_i)$ before placing $\cp(C_i)$.
We identify the first maximal clique $C_j, j < i$ that is
placed to the right of all points in $\Reg(C_i)$. Since $\cp(C_j)\notin \Reg(C_i)$, 
we have that $\Pre(C_i)\ne\Pre(C_j)$.
By Lemma~\ref{lem:singleGap}, the maximal clique
$C_j$ has a gap $J$ with $\Reg(C_i)\subseteq J$, see Fig.~\ref{fig:nhca_char}.
    
Consider the neighboring island $I$ of $C_j$ to the left of $J$. Since
$\cp(C_j)$ was not placed on $I$, the clique point $\cp(C_{j-1})$ has
been placed to the right of $I$. By the choice of $C_j$, we have that
$\cp(C_{j-1})$ is not placed to the right of $J$ and thus $\cp(C_{j-1})\in J$. 
Since $\cp(C_j)$ has been placed to the right of $\Reg(C_i)$ and thus to
the right of $J$, we have $p_D\not\in J$ and thus $D\not\in S_J$.
With $D<C_{j-1}< C_j< C_i$, where $C_{j-1}, C_i\in S_J$ and $D,C_j\not\in S_J$
we get a contradiction with the property~\ref{itm:HlinCh4}, since $S_J$ is not consecutive in $[<]$.

From the placement of clique points, 
we obtain an $\hca$-extension of $\calR'$ so that for every not yet represented 
vertex $u \notin V(G')$ we choose $R(u)$ to be a minimal arc containing 
exactly the clique points of the maximal cliques from $M_u$.

By the Lemma~\ref{lem:gavril_hca}, we obtain from
properties~\ref{itm:HlinCh1},~\ref{itm:HlinCh4} that this results in an $\hca$ representation of $G$
extending $\calR'$, since the maximal clique points are placed
correctly with regards to both predrawn arcs and new arcs, and moreover, two arcs 
intersect if and only if they are predrawn or share a maximal clique.
Note that we can perturbate the endpoints of the non-predrawn arcs to avoid
shared endpoints.
\end{proof}

\begin{algorithm}[t!]
  \SetKw{Continue}{continue} 
  \KwData{A graph $G$ and a partial representation $\calR'$ consisting of arcs with pairwise distinct ends.}
  \KwResult{An \hca{} representation $\calR$ of $G$ extending $\calR'$ or a message that it does not exist.}
  \Begin{
    Determine the set of maximal cliques $\calC$ of $G$\;
    \ForEach{$C\in \calC$}{
      Determine $\Reg(C)$\;
      \lIf{$\Reg(C)=\emptyset$}{\Return{$\calR'$ has no extension}}
    }
    $\varepsilon:=\frac{1}{2n+1}\min\{|I|,I \text{ is an island}\}$\;
    Choose any $C_1\in\calC$\;
    \ForEach{island $I$ of $C_1$}{
      Set $\cp(C_1)$ to the middle of $I$\;
      Determine the partial order $\prec$\;
      Build a PC tree $T$ on $\calC$ capturing the constraints stated in  Theorem~\ref{the:HCArepextChar}\;
      \lIf{such $T$ does not exist}{\Continue}
      Solve $\reorder(T, C_1, \prec)$ to get the order $C_1<C_2<\dots<C_k$\;
      \lIf{such order $<$ exists}{
        \For{$i=2$ \KwTo $k$}{ 
          $P:=\Reg(C_i) \cap \{p: p> \cp(C_{i-1})+\varepsilon\}$\;
          \lIf {$\min P$ exists}{$\cp(C_i):=\min P$}
          \lElse{$\cp(C_i):=\inf(P)+\varepsilon$}
        }
        \lForEach{$u \notin V(G')$}{draw $R(u)$ to cover exactly $\{\cp(C): C\in M_u\}$} 
        \Return{$\calR$}
      }
    }
  }
  \Return{$\calR'$ has no extension}\;
  \caption{The algorithm for the $\RepExt(\hcar)$ problem.}\label{alg:hca}
\end{algorithm}

\hcapoly*
\begin{proof}
  The correctness of Algorithm~\ref{alg:hca} follows from the already mentioned arguments.
  For the computational complexity of the more complex steps note that:
  \begin{compactitem}
    \item
      Line 2:We run the linear time recognition algorithm for Helly
      circular-arc graphs~\cite{lin2006characterizations} on $G$ and read its
      at most $n$ maximal cliques from any of its representation.
    \item
      Lines 3--5: The regions can be obtained in time $\calO(n)$ by traversing
      the circle once.  During the traversal two consecutive predrawn
      endpoints specify possible islands which can be assigned to appropriate
      maximal cliques.
    \item
      Line 8: There are at most $O(n)$ islands for $C_1$.
    \item 
      Line 10: The comparable pairs of the partial order $\prec$ can also be
      determined by a traversal as in Steps 3--5.
    \item
      Line 11: The construction of the PC-tree follows the standard approach
      from the recognition of Helly circular-arc graphs~\cite{hsu2003pc}. 
      Each arc is at the left start of at most one gap and each of the $O(n)$
      maximal cliques has at most one gap without an arc inside. Hence, there
      are $O(n)$ gaps and thus $O(n)$ constraints. $T$ can thus be constructed
      in $O(n^2)$ time~\cite{hsu2003pc}.
    \item
      Line 13: The $\reorder(T, C_1, \prec)$ problem can be solved in
      $\calO(n^2)$ time by Lemma~\ref{lem:reordering_linear}.
  \end{compactitem}
\end{proof}


\section{Conclusions and Open Problems}
\label{sec:conclusions}

Our study of the $\RepExt$ problem has been restricted in two ways:

First, we have considered mostly representations satisfying Helly property as
this allows us to consider the clique points of maximal cliques. For
representations that do not have this property one would involve a completely
different approach.

Secondly, for the recognition problem it is irrelevant whether arcs are closed
or open, but this is not the case for the representation extension. Observe
that touching intervals in $\RepExt(\nphcar)$ in Lemma~\ref{lem:nphca_char}
imply constraints on the ordering. For the sake of completeness it might be
worth to check whether use of open or semi-open intervals would yield
significant impact on the computational complexity.

\paragraph*{Acknowledgement.}
\label{sec:acknowledgement}
We thank Bartosz Walczak for inspiring comments, in particular for his hint to
extend Theorem~\ref{thm:hca_npc} to the case of $\car$ with distinct endpoints.



\bibliography{repext_ca}
\end{document}